\documentclass[12pt,letterpaper]{article}
\usepackage[letterpaper]{geometry}
\usepackage{amssymb,amsmath}
\usepackage{mathtools}
\usepackage{amsthm}
\usepackage{bbm}
\usepackage{graphicx}
\usepackage{enumerate}
\usepackage{caption}
\usepackage{natbib}
\usepackage{url} 
\usepackage{bm}
\usepackage{color}
\usepackage{silence} % To get rid of the following warnings:
\newcommand{\indep}{\perp \!\!\! \perp}

\newcommand{\Var}{\mbox{Var}}
\newcommand{\cov}{\mbox{Cov}}
\newcommand{\Cov}{\mbox{Cov}}
\newcommand{\cor}{\mbox{Cor}}
\newcommand{\Cor}{\mbox{Cor}}
\newcommand{\dCov}{\mbox{dCov}}

\newcommand{\dCor}{\mbox{dCor}}
\newcommand{\E}{\mathbb{E}}

\newcommand{\CD}{{\boldsymbol{\cdot}}}
\newcommand{\p}{p}
\newcommand{\rp}{\tau}

\newtheoremstyle{break}
  {\topsep}{\topsep}
  {\itshape}{}
  {\bfseries}{}
\theoremstyle{break}

\newtheorem{lemma}{Lemma}
\newtheorem{proposition}{Proposition}

\usepackage{xcolor}
\definecolor{blue}{RGB}{0,0,255}
\definecolor{red}{RGB}{255,0,0}

\newcommand{\blind}{0}

\begin{document}

\def\spacingset#1{\renewcommand{\baselinestretch}%
{#1}\small\normalsize} \spacingset{1}

%========================================================

\if0\blind
{
  \title{\bf Distance Covariance, Independence,\\ and 
       Pairwise Differences\vspace{3mm} }       
  \author{Jakob Raymaekers\\
    Department of Mathematics, University of Antwerp, 
    Belgium\\ \\
    %and\\ 
    Peter J. Rousseeuw\\
    Section on Statistics and Data Science, KU Leuven, Belgium\\
    \\ June 18, 2024 \\ \\ \\
    (To appear in {\it The American Statistician}.)\\}  
  \date{} %June 16, 2024\\ %\today}
  \maketitle
} \fi

\if1\blind
{
  \phantom{abc}\\
  \bigskip
  \bigskip
  \begin{center}
    {\LARGE\bf Distance Covariance, 
    Independence,\\ \vskip2mm and 
       Pairwise Differences\\}
  \bigskip \bigskip
  \large{\today}\\
  \bigskip
  \end{center}
} \fi

\begin{abstract}
Distance covariance \citep{Szekely2007} is a 
fascinating recent notion, which is popular as a test 
for dependence of any type between random variables $X$
and $Y$. This approach deserves to be touched upon in
modern courses on mathematical statistics. It makes
use of distances of the type $|X-X'|$ and $|Y-Y'|$,
where $(X',Y')$ is an independent copy of $(X,Y)$. 
This raises natural questions about independence 
of variables like $X-X'$ and $Y-Y'$, about the
connection between $\cov(|X-X'|,|Y-Y'|)$ and the
covariance between doubly centered distances, and
about necessary and sufficient conditions for
independence. We show some basic results and 
present a new and nontechnical counterexample to 
a common fallacy, which provides more insight. 
We also show some motivating examples involving 
bivariate distributions and contingency tables,
which can be used as didactic material for 
introducing distance correlation.
\end{abstract}


%==============================================
\newpage
\spacingset{1.45} % DON'T change the spacing!

\section{Introduction} \label{sec:intro}

Independence of random variables is an
important and nontrivial topic in probability
and statistics. There are many subtleties 
concerning independence and correlation, see 
e.g. \cite{Muk2022}, \cite{Rodgers1988}, and
\cite{Rousseeuw1994}.
It is often emphasized in class that two
real random variables $X$ and $Y$ having zero
covariance does not imply their independence.
The recent work of \cite{Szekely2007} provided
a surprising contrast, since the distance
covariance they introduced {\it does}
characterize independence. In our opinion this
topic would be a valuable addition to a
graduate course on mathematical statistics,
because distance covariance is a general method
with interesting properties and wide ranging
applications, for instance in variable 
selection \citep{Chen2018}, sparse contingency
tables \citep{Zhang2019}, independent 
component analysis \citep{Matteson2017}, and 
time series \citep{Davis2018}.
It can be computed fast 
\citep{Huo2016,Chaudhuri2019}, 
and there are interesting connections with other
dependence measures \citep{Edelmann2022}. 
Its robustness to outliers was studied recently
\citep{Leyder2024}.

The formulation of the distance covariance,
described in Section~\ref{sec:counterexample},
is very simple but contains some subtleties that
often give rise to misunderstandings.
It is based on pairwise differences $X-X'$ and 
$Y-Y'$, where $(X',Y')$ is an independent copy 
of $(X,Y)$. In order to provide a context 
for the role of these pairwise differences, we
establish some connections between independence 
of $X$ and $Y$ and independence relations 
involving $X-X'$ and $Y-Y'$. We have not found 
these results in the literature, and we believe 
they could provide a pedagogic background.

We also construct an elementary counterexample to
a common misunderstanding, with the aim of
clarifying why the distance covariance approach
requires ``double centering'' of the interpoint
distances $|X-X'|$ and $|Y-Y'|$.

Most of the material in this paper is accessible
to students who took an introductory course in
probability and statistics. Only the statements 
of Proposition~\ref{prop:XminX'indepY}(b) and 
Proposition~\ref{prop:XminX'indepYminY'}(c) and
the proofs in the Appendix require knowledge
of characteristic functions, but this is not
needed to follow the examples.

%================================================
\section{Some results on pairwise differences} 
\label{sec:results}

Let us denote independence of a pair of real 
random variables as $X \indep Y$.
We start by looking at pairwise differences of
only one of the variables, say $X$. We consider 
an independent copy $X'$ of $X$, that is,
$X' \sim X$ and $X' \indep (X,Y)$. Then the 
following implications hold.

\begin{proposition} \label{prop:XminX'indepY}
For a pair of random variables $(X,Y)$ 
it holds that\\
{\it (a)} $X \indep Y$ implies $(X-X') \indep Y$\,.\\ 
{\it (b)} If the characteristic function of $X$ 
has no roots or only isolated roots, or the\\ 
\phantom{{\it (b)}} characteristic function of 
 $(X,Y)$ is analytic, then
 $(X-X') \indep Y$ implies $X \indep Y$\,. 
\end{proposition}

The proof can be found in the Appendix.
Part (a) is general, as it does not require 
any conditions on $X$ or $Y$, such as the 
existence of certain moments. Part (b) is
a bit more involved. 
We have been unable to find this proposition 
in the literature, but since part (a) is 
straightforward we expect that it is known.

The conditions on the characteristic functions
in part (b) of 
Proposition~\ref{prop:XminX'indepY} 
look quite stringent, but there are many relevant 
cases. The characteristic functions of the 
Gaussian, Student, exponential, Poisson, 
chi-square, Gamma, Laplace, logistic, Cauchy, 
and stable distributions have no roots. 
Distributions whose characteristic functions have
non-isolated zeroes are unusual, but some examples
do exist, see e.g. \citet{Ushakov1999}, page 265.
The alternative condition that $\phi_{(X,Y)}$ is 
analytic is satisfied whenever $X$ and $Y$ are 
bounded, see e.g. \citet{Berezin2016}, page 147.

Next we consider pairwise differences of both
$X$ and $Y$. For this we take an independent 
copy $(X',Y')$ of $(X,Y)$, that is,  
$(X',Y') \sim (X,Y)$ and $(X',Y') \indep (X,Y)$.

\begin{proposition} \label{prop:XminX'indepYminY'}
For a pair of random variables $(X,Y)$ it 
holds that\\
{\it (a)} $X \indep Y$ implies 
   $(X-X') \indep (Y-Y')$\,.\\ 
{\it (b)} $(X-X') \indep (Y-Y')$ and 
   $(X-X') \indep Y$ together imply $X \indep Y$.\\
{\it (c)} If $(X,Y)$ is symmetric and its 
  characteristic function has no roots or 
  is analytic,\\
\phantom{{\it (c)}} $(X-X') \indep (Y-Y')$ 
  implies $X \indep Y$.
\end{proposition}

We could not find these results
in the literature either, and in our opinion
they could provide a useful background when 
the notion of distance covariance is taught.
Also, parts (a) of Propositions 
\ref{prop:XminX'indepY} and 
\ref{prop:XminX'indepYminY'}
could be used as exercises in a
chapter on characteristic functions.
Together with the partial converses in these
propositions they would make a viable 
homework assignment, as long as the exact 
statements of the propositions are provided,
and perhaps also those of the lemmas in the 
Appendix.

It is worth noting that the converse of part
(a) of Proposition~\ref{prop:XminX'indepYminY'}
does not hold without further conditions, 
because there exists a nontrivial counterexample 
(Gabor Sz\'ekely 2024, personal communication).
Therefore also the converse of 
Proposition~\ref{prop:XminX'indepY}(a) cannot 
hold without further conditions, or else we 
could prove the converse of 
Proposition~\ref{prop:XminX'indepYminY'}(a) by 
applying the converse of 
Proposition~\ref{prop:XminX'indepY}(a) twice.

%================================================
\section{Connection with distance covariance}
\label{sec:counterexample}

If $X \indep Y$ we obtain from 
Proposition~\ref{prop:XminX'indepYminY'}(a) that 
$(X-X') \indep (Y-Y')$. But then it follows that 
also\linebreak
\mbox{$|X-X'| \indep |Y-Y'|$}, since the absolute 
value is a continuous function. 
If $X$ and $Y$ have second moments, that is, $\E[X^2]$
and $\E[Y^2]$ are finite, also $\E[|X-X'|^2]$ and
$\E[|Y-Y'|^2]$ are finite. Therefore the covariance
of $|X-X'|$ and $|Y-Y'|$ exists as well, and since 
$|X-X'| \indep |Y-Y'|$ we have
\begin{equation} \label{eq:covdist}
   \cov(|X-X'|,|Y-Y'|) = 0.
\end{equation}
Therefore, when the second moments of $X$ and
$Y$ exist, $\cov(|X-X'|,|Y-Y'|) = 0$ is a 
{\it necessary} condition for $X \indep Y$. 
However, it is not a {\it sufficient} condition.
In order to illustrate this, we set out to 
construct a simple counterexample. 

\vspace{4mm}
\noindent {\bf Example 1.} The smallest
example we were able to produce is a probability 
distribution on 4 points in the plane.
Table~\ref{tab:counterexample} lists the coordinates 
of $(X,Y)$, and the 4 points are plotted in the 
top panel of Figure~\ref{fig:4points}. 
Note that $X$ and $Y$ are uncorrelated but not 
independent, since the distribution of $Y|X=x$ 
depends on $x$. The $Y \sim X$ regression line is
horizontal. 
The resulting distribution of $(|X-X'|,|Y-Y'|)$ 
contains 5 points, given in the middle panel 
of Table~\ref{tab:counterexample}
with their probabilities, and plotted in the
bottom left panel of Figure~\ref{fig:4points}.
It is easily verified that
$\cov(|X-X'|,|Y-Y'|)$ is exactly zero.
So this is an example with non-independent $X$ and
$Y$ for which $\cov(|X-X'|,|Y-Y'|)=0.$

\begin{table}[ht]
\centering
\caption{Example 1: $\cov(|X-X'|,|Y-Y'|)=0$
but $X$ and $Y$ are dependent.}
\label{tab:counterexample}
\vspace{-0.2cm}
\begin{tabular}{rrrc}
\hline
 atom  & X  & Y & probability \\    
\hline
 $a$ & -1 & 1 & 1/4 \\ 
 $b$ &  1 & 1 & 1/4 \\ 
 $c$ &  0 & 0.6 & 1/4 \\ 
 $d$ &  0 & -1 & 1/4\\ 
\hline
  atom & $|X-X'|$ & $|Y-Y'|$ & probability \\    
\hline
$(a,a)\,,(b,b)\,,(c,c)\,,(d,d)$ & 0 & 0 & 1/4 \\ 
$(c,d)$ & 0 & 1.6 & 1/8 \\ 
$(a,c)\,,(b,c)$ & 1 & 0.4 & 1/4 \\ 
$(a,d)\,,(b,d)$ & 1 & 2 & 1/4 \\ 
$(a,b)$ & 2 & 0 & 1/8 \\ 
\hline
  atom & $\Delta(X,X')$ & $\Delta(Y,Y')$
  & probability \\    
\hline
$(a,a)\,,(b,b)$ & -1.25 & -0.40 & 1/8 \\ 
$(d,d)$ & -0.25 & -2.00 & 1/16 \\ 
$(c,c)$ & -0.25 & -0.40 & 1/16 \\ 
$(c,d)$ & -0.25 & 0.40 & 1/8 \\ 
$(a,c)\,,(b,c)$ & 0.25 & 0.00 & 1/4 \\ 
$(a,d)\,,(b,d)$ & 0.25 & 0.80 & 1/4 \\ 
$(a,b)$ & 0.75 & -0.40 & 1/8 \\ 
\hline
\end{tabular}
\label{tab:4points}
\end{table}

\begin{figure}[!ht]
\centering
\includegraphics[width=.49\textwidth]
   {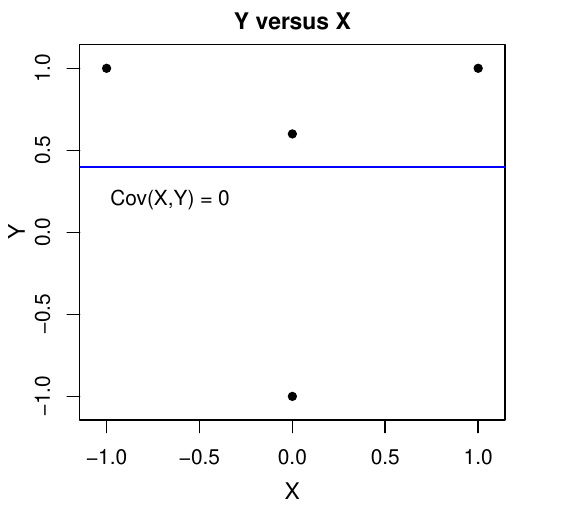}
\includegraphics[width=.98\textwidth]
   {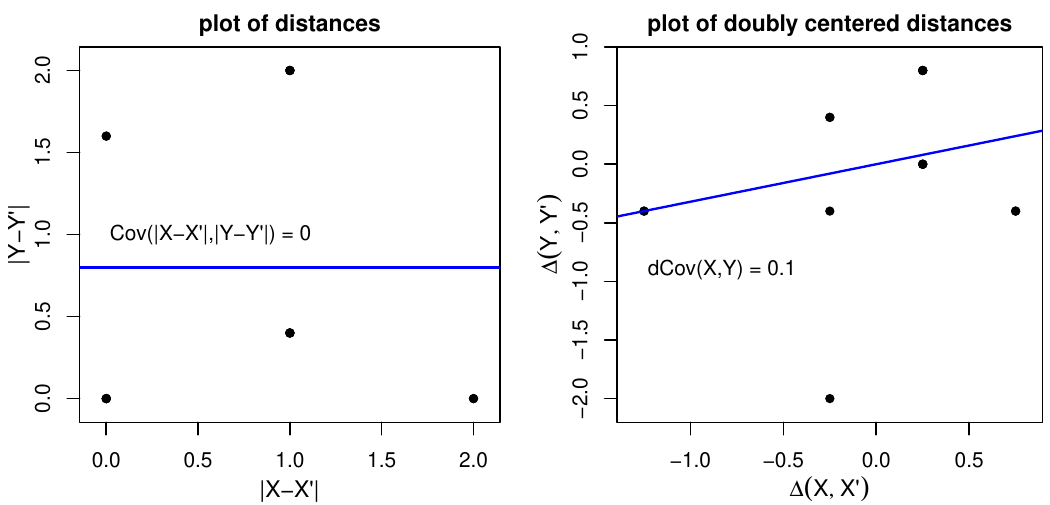}
\vspace{2mm}
\caption{Example with a distribution on 4 points, 
  from Table~\ref{tab:4points}.
  Top: plot of $Y$ versus $X$. 
  Bottom left: plot of pairwise distances $|Y-Y'|$
  of $Y$ versus those of $X$. Bottom right: doubly
  centered distances $\Delta(Y,Y')$ of Y versus 
  those of $X$.}
  \label{fig:4points}
\end{figure}

\vspace{4mm}
\citet{Szekely2007} proposed to use another function. 
Instead of the interpoint distances $|X-X'|$
above, they compute their {\it doubly centered} 
version given by
\begin{align} \label{eq:DC}
 \Delta(X,X') =&\, |X-X'| - \E_{X''}[|X-X''|]
 \nonumber\\
  &\,-\E_{X''}[|X''-X'|]
  + \E_{X'',X'''}\E[|X''-X'''|]
\end{align}
where $X''$ and $X'''$ are also independent 
copies of $X$. For $\Delta(X,X')$ to exist it is 
necessary that  $E[|X|]$ is finite. Note that 
$\Delta(X,X') = \Delta(X',X)$ is not a distance 
itself, since it also takes on negative values.
Moreover, $\E_{X}[\Delta(X,X')]$ is zero, and the 
same holds for $\E_{X'}[\Delta(X,X')]$ and 
$\E_{X,X'}[\Delta(X,X')]$.
This explains the name `doubly centered'. It turns 
out that the {\it second} moments of 
$\Delta(X,X')$ exist as well.

If also $\E[|Y|]$ is finite, \citet{Szekely2007}
compute what they call the {\it distance covariance}
of $X$ and $Y$, given by
\begin{equation} \label{eq:dCov}
  \dCov(X,Y) := \cov(\Delta(X,X'),\Delta(Y,Y')).
\end{equation}
(In fact they took the square root of the right 
hand side, but we prefer not to because the units 
of~\eqref{eq:dCov} are those of $X$ times $Y$.)
They proved the amazing result that when the first 
moments of $X$ and $Y$ exist, it holds that
\begin{equation} \label{eq:dCovzero}
  X \indep Y \Longleftrightarrow \dCov(X,Y)=0.
\end{equation}
This yields a necessary {\it and sufficient} 
condition for independence. The implication 
$\Longleftarrow$ is not obvious at all, and was 
proved by complex analysis. Their 
work also made it clear that always 
$\dCov(X,Y)\geqslant 0$ because they can write 
$\dCov(X,Y)$ as an integral of a nonnegative 
function.

The bottom panel of 
Table~\ref{tab:counterexample} 
lists the coordinates of the  
$(\Delta(X,X'),\Delta(Y,Y'))$ and their 
probabilities, and these points are plotted in the 
bottom right panel of Figure~\ref{fig:4points}.
Note that we now have 7 points instead of 5.
Indeed, the atom 
$\{(a,a)\,,(b,b)\,,(c,c)\,,(d,d)\}$
of $(|X-X'|,|Y-Y'|)$ has split into three atoms
of $(\Delta(X,X'),\Delta(Y,Y'))$. Even though
all four pairs have the same 
$(X-X',Y-Y') = (0,0)$, they can obtain different 
$(\Delta(X,X'),\Delta(Y,Y'))$ because different
means were subtracted from their $|X-X'|=0$ 
and $|Y-Y'|=0$. This implies that the doubly
centered distance $\Delta(X,X')$ cannot be
written as a function of $X-X'$.

In spite of the name `distance covariance', $\dCov$ 
is thus fundamentally different from the covariance 
of distances in~\eqref{eq:covdist}.
As we just saw, $\dCov$ is not a function of the 
pairwise differences $X-X'$ and $Y-Y'$ alone: to
compute $\Delta(x,x')$ we need to know the actual 
values of $x$ and $x'$. 
So the arrow $\Longrightarrow$ 
in~\eqref{eq:dCovzero} is not an immediate 
consequence of $|X-X'| \indep |Y-Y'|$
or even of the fact that $X-X' \indep Y-Y'$\,, 
instead it is truly derived from $X \indep Y$.
(If $\Delta(X,X')$ were a function of $X-X'$
it would follow from $(X-X') \indep (Y-Y')$
that $\dCov(X,Y)=0$ and hence $X \indep Y$,
which we know is not true in general.)

In the example we obtain exactly 
$\dCov(X,Y) = 0.1 > 0$, which confirms the 
dependence of $X$ and $Y$. 
The example thus illustrates that the 
double centering in $\dCov(X,Y)$ is necessary
to characterize independence, since without 
it we obtained $\cov(|X-X'|,|Y-Y'|)=0$
which provided no clue about the dependence 
of $X$ and $Y$.

The regression line
in the bottom right panel of 
Figure~\ref{fig:4points} is not 
horizontal but goes up. Its slope must 
be positive or zero because it is a positive 
multiple of $\dCov(X,Y)$, which we know is 
always nonnegative.
The regression line also has to pass through
the origin $(0,0)$, because the doubly 
centered distances of X as well as Y have 
zero mean, so the average of the points in 
this plot is the origin.
In this tiny example the regression line also 
happens to pass through one of the points in 
the plot, but that is a coincidence. The line 
does not have to pass through any  
point, as can be verified by e.g. changing 
the first x-coordinate of the original data 
from -1.0 to -1.5\,.

\citet{Szekely2007} also derived a different 
expression for $\dCov$. Working out the 
covariance in~\eqref{eq:dCov} yields
$4 \times 4 = 16$ terms, that exist when $X$ and 
$Y$ also have second moments. With elementary 
manipulations and a lot of patience these terms 
can be reduced to three:
\begin{align*} \label{eq:||}
  \dCov(X, Y)
  =&\, \E[|X-X'||Y-Y'|]+\E[|X-X'|]\E[|Y-Y'|]
  \nonumber \\
  &\, -2\E[|X-X'||Y-Y''|].
\end{align*}
Combining the first term on the right with minus the 
second, and the third with twice the second, 
\cite{SzekelyRizzo2023} obtain
\begin{equation*}
    \dCov(X,Y) = \cov(|X-X'|,|Y-Y'|)
    - 2 \cov(|X-X'|,|Y-Y''|)
\end{equation*}
which connects $\dCov$ with the covariance of
distances in \eqref{eq:covdist}.
Since we have seen that $X \indep Y$ implies that
$\cov(|X-X'|,|Y-Y'|) = 0$, the
only way that $X$ and $Y$ can be independent is
when {\it both} terms on the right hand side are 
zero. In the example $\cov(|X-X'|,|Y-Y'|) = 0$ 
but $X$ and $Y$ are dependent, so the second term 
has to be nonzero, and indeed  
$\cov(|X-X'|,|Y-Y''|) = -0.05$\,.
 
%==============================================
\section{Distance correlation and
         finite samples} \label{sec:finite}

Since the units of $\dCov(X,Y)$ 
are those of $X$ times $Y$, and 
$\dCov(aX,bY) = ab\,\dCov(X,Y)$,
one often uses the unitless {\it distance
correlation} defined as
\begin{equation}\label{dcor}
  \dCor(X,Y) = \frac{\dCov(X,Y)}
  {\sqrt{\dCov(X,X)\dCov(Y,Y)}}\;
\end{equation}
which always lies between 0 and 1.
Note that the conventional definition is
the square root of \eqref{dcor}.

So far we have worked with population
distributions, but $\dCov$ and $\dCor$ can also
be used for finite samples. One can simply apply
them to the empirical distribution of the
sample. In particular, for a univariate
sample $X_n = (x_1,\ldots,x_n)$ we denote
$d_{ij} := |x_i-x_j|$ for $i,j=1,\ldots,n$
as well as 
\begin{equation}
\overline{d_{i \CD}} = 
  \frac{1}{n}\sum_{j=1}^n d_{ij} \qquad
\overline{d_{\CD j}} = 
  \frac{1}{n}\sum_{i=1}^n d_{ij} \qquad
\overline{d_{\CD \CD}} =
  \frac{1}{n^2}\sum_{i,j=1}^n d_{ij}\;\;.
\end{equation}
Double centering yields the values
$$\Delta^{\!X_n}_{ij} := d_{ij} - 
  \overline{d_{i \CD}} - 
  \overline{d_{\CD j}} + 
  \overline{d_{\CD \CD}}$$
so that 
$\sum_{j=1}^n \Delta^{\!X_n}_{ij} = 0$ 
for all $i$ and
$\sum_{i=1}^n \Delta^{\!X_n}_{ij} = 0$ 
for all $j$.
The $\dCov$ of a bivariate sample is then
defined as
\begin{equation}\label{dCovsample}
  \dCov(X_n,Y_n) = \frac{1}{n^2} 
  \sum_{i,j=1}^n \Delta^{\!X_n}_{ij} 
  \Delta^{\!Y_n}_{ij}\;\;.
\end{equation}
The $\dCor$ of a bivariate sample is 
analogous to~\eqref{dcor}.
When based on an i.i.d. sample of size $n$
from a pair of random variables $(X,Y)$ 
with first moments, the finite-sample 
$\dCov(X_n,Y_n)$ converges almost surely to 
$\dCov(X,Y)$ when $n \rightarrow \infty$
\citep{Szekely2007}.

\section{Examples}
\label{sec:examples}

The material in this section and the 
next one can be used as exercises for 
students, in a lab session or a 
homework assignment.

\noindent {\bf Example 2.} The 
distance covariance can be applied to 
contingency tables. 
For instance, $2\times 2$ contingency 
tables can be modeled by Bernoulli
variables $X$ and $Y$, that can only take 
on the values 0 and 1. We denote their 
joint probability as $\p_{ij}=P(X=i,Y=j)$
and the marginal probabilities as
$\p_{i\CD} = \p_{i0}+\p_{i1}$ and
$\p_{\CD j}= \p_{0j}+\p_{1j}$\,.
It can be verified that
\begin{equation}\label{dCov_contig}
  \dCov(X,Y) = \sum_{i=0}^1 \sum_{j=0}^1 
  (\p_{ij} - \p_{i  \CD}\p_{\CD j})^2\,.
\end{equation}
Therefore $\dCov(X,Y)=0$ iff 
$\p_{ij} = \p_{i\CD}\p_{\CD j}$ for all
$i,j = 0,1$, which is equivalent to 
$X \indep Y$.
Note that~\eqref{dCov_contig} is similar 
to Pearson's chi-square statistic, but not 
identical. If we divide the chi-square
statistic by the sample size, and let the
sample size grow, it converges
to the population version 
\begin{equation*}\label{chisquare}
  \sum_{i=0}^1 \sum_{j=0}^1 
  \Big(\frac{\p_{ij} - 
  \p_{i  \CD}\p_{\CD j}}
  {\p_{i  \CD}\p_{\CD j}}\Big)^2
\end{equation*}
which is not equivalent to 
\eqref{dCov_contig}.
It is not too difficult to derive that  
\begin{equation}\label{covD_contig}
  \cov(|X-X'|,|Y-Y'|) = 2\big(\p_{00}\p_{11} +
  \p_{01}\p_{10} - 
  2\p_{0\CD}\p_{1\CD}\p_{\CD 0}\p_{\CD 1}
  \big) \;.
\end{equation}
Now it is easy to see that $X \indep Y$
implies that \eqref{covD_contig} becomes zero. 
But it is not true the other way around. A 
counterexample is given by 
$(\p_{00},\p_{01},\p_{10},\p_{11}) =
(10,5,14,11)/40$.
This zeroes $\cov(|X-X'|,|Y-Y'|)$,
but $X$ and $Y$ are not independent and 
$\dCov(X,Y) = 0.025$ is strictly positive.
(Unlike Example 1 in Table~\ref{tab:4points},
here the plain $\cov(X,Y)$ is not zero.)\\

\noindent {\bf Example 3.} The main 
advantage of $\dCor$ over the usual
product-moment correlation $\Cor$ is that from
$\dCor(X,Y)=0$ it follows that $X \indep Y$.
Most introductory statistics books stress that
this does not hold for $\Cor$. A typical 
illustrative example is to
take a univariate variable $X$ with a 
distribution that has a second moment and is 
symmetric about zero, and to put $Y = X^2$.
(If a bivariate density is desired,
one can add a Gaussian error term to $Y$.)
Let us take the simple case where $X$ follows
the uniform distribution on $[-1,1]$.
Clearly $X$ and $Y$ are dependent, but by
symmetry $\Cor(X,Y)=0$. However, we will see
that $\dCor(X,Y)$ is strictly positive.

The computation of $\dCor(X,Y)$ offers 
an opportunity for carrying out a simple
numerical experiment. 
First we have to generate a sample of size 
$n$ from this bivariate distribution. 
This is easy, for instance in \textsf{R}
we can run 
\texttt{X=runif(n,min=-1,max=1)} 
followed by \verb+Y=X^2+\,. 
This yields the left panel of 
Figure~\ref{fig:examples_3_and_4}, in
which the horizontal regression line
illustrates that the classical 
correlation $\cor(X,Y)$ is zero. To 
compute the sample distance correlation 
we can use the \texttt{R} package 
\texttt{energy} \citep{Rizzo2022} or 
the package \texttt{dccpp} 
\citep{Berrisch2023}. In the first case 
we run \texttt{energy::dcor2d(X,Y)} which 
uses the algorithm of \cite{Huo2016}, 
and in the second case the command is
\verb+dccpp::dcor(X,Y)^2+ 
which carries
out the algorithm of \cite{Chaudhuri2019}.
Both algorithms for $\dCor$ are very fast 
as their computation time is only 
$O(n\,\log(n))$, and they do not store the 
$n \times n$ matrices of all 
$\Delta^{\!X_n}_{ij}$ and 
$\Delta^{\!Y_n}_{ij}$\,.
When we let $n$ grow, the answer quickly
converges to approximately $0.2415 > 0$.
The result stabilizes even faster if
we use an equispaced set
\texttt{X=seq(from=-1,to=1,by=2/(n-1))},
so the computation becomes a crude
numerical integration.\\

\begin{figure}[!ht]
\centering
\includegraphics[width=.98\textwidth]
   {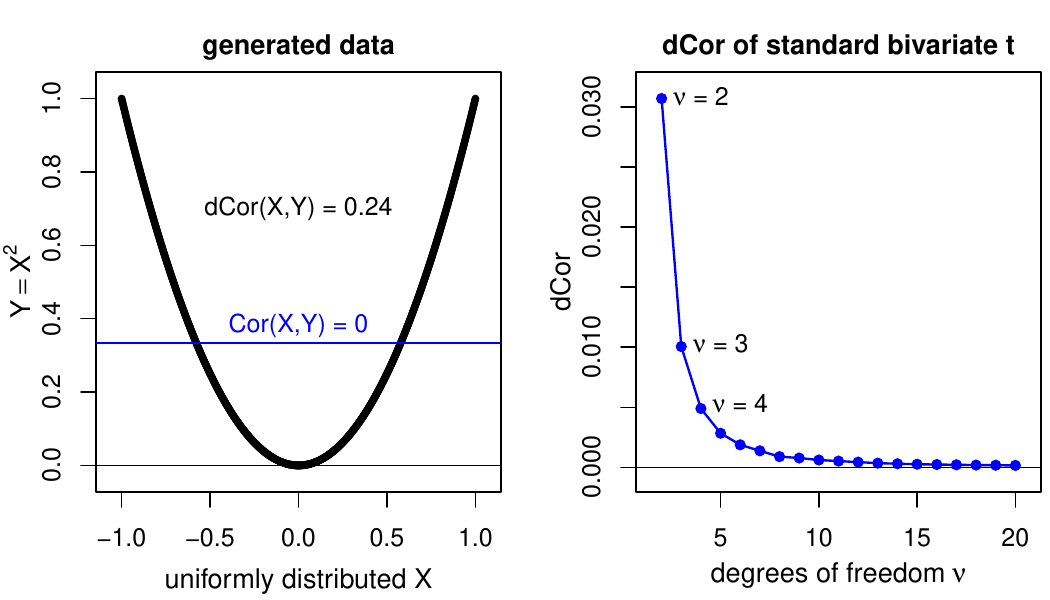}
\vspace{2mm}
\caption{Left: dependent variables 
 generated 
 in Example 3, with horizontal regression 
 line illustrating that $X$ and $Y$ are 
 uncorrelated. Right: Plot of the distance
 correlation of the standard bivariate
 $t$-distribution in Example 4, for a range 
 of $\nu$.}
\label{fig:examples_3_and_4}
\end{figure}

\noindent {\bf Example 4.} In the 
previous example the left panel of
Figure~\ref{fig:examples_3_and_4} 
immediately reveals the dependence, because 
the conditional 
expectation $E[Y|X=x] = x^2$ depends on $x$. 
But there are more subtle situations, where 
for instance the conditional expectation is 
constant but some other moment is not. A nice 
example is the bivariate $t$-distribution.
When its center is $(0,0)$ and its scatter
matrix is the identity matrix, it is called
the standard bivariate $t$-distribution
with density
\begin{equation}\label{bivart}
  f(x,y) = \frac{1}{2\pi} \Big(1 + 
  \frac{x^2+y^2}{\nu}\Big)^{-(\nu+2)/2}
\end{equation}
where $\nu$ is the degrees of freedom
parameter. The marginal distribution of $Y$ is
the usual univariate $t$-distribution with
center $0$ given by
\begin{equation}\label{univart}
  f_t(y\,;s^2,\nu) =
  \frac{c(\nu)}{s}\Big(1+\frac{(y/s)^2}{\nu}
  \Big)^{-(\nu+1)/2}
\end{equation}
where $c(\nu)$ is the constant needed to make
the density integrate to 1, and the scale 
parameter $s$ equals 1 here. In general 
$\Var(Y) = s^2\nu/(\nu-2)$ when $\nu>2$.
A plot of the bivariate density \eqref{bivart} 
looks a lot like that of the standard 
bivariate Gaussian distribution, with circular 
symmetry. When $\nu > 2$ the correlation
$\cor(X,Y)$ exists and is zero.
But whereas $X$ and $Y$ are 
independent in the standard Gaussian 
setting, they are no longer here, since
the bivariate density \eqref{bivart} does 
not equal the product of the marginal
densities of $X$ and $Y$.
The conditional density of $Y$ given 
$X=x$ is now
\begin{equation}\label{conditionalf}
 f(y|X=x) = f_t\Big(y\,;\frac{\nu+x^2}
 {\nu+1}, \nu+1\Big)
\end{equation}
\citep{Ding2016}, so it is again a 
univariate $t$ with center 0, but now 
with $\nu + 1$ degrees of freedom and 
a scale parameter that depends on $x$. 
Due to the 
increased degrees of freedom, the
conditional expectation already exists 
for $\nu > 0$ and equals zero, so it is 
constant. The conditional variance 
exists for any $\nu > 1$ and equals 
$(\nu+x^2)/(\nu-1)$. It is thus lowest 
for $x=0$ and increases with $|x|$.

We now study the distance 
correlation of these dependent but 
uncorrelated variables $X$ and $Y$.
An analytic derivation of $\dCov(X,Y)$
may not be possible, but in \textsf{R}
we can easily generate data from the 
standard bivariate $t$-distribution by
\texttt{rmvt(n,df=df)} where \texttt{df}
is the degrees of freedom $\nu$. The 
function \texttt{rmvt} is in the \texttt{R}
package \texttt{mvtnorm} \citep{Genz2023}.
We can then compute the distance 
correlation in exactly the same way as in
Example 3 above. 
Figure~\ref{fig:examples_3_and_4} shows
the resulting estimates of $\dCor(X,Y)$
obtained for $n=$100 000 and $\nu$ 
ranging from 2 to 20, a computation that 
took under one minute.
The distance correlation goes down to
zero for increasing $\nu$, which is
understandable because for 
$\nu \rightarrow \infty$ the standard
bivariate $t$-distribution converges
to the standard Gaussian distribution,
where $X$ and $Y$ are indeed independent.

%==========================================
\section{Testing for independence}
\label{sec:permtest}

Now suppose we have an i.i.d. 
sample $(X_n,Y_n)$ from a bivariate random
variable $(X,Y)$, and we want to test the
null hypothesis $H_0$ that $X$ and $Y$ are
independent.
If we know that $(X,Y)$ is bivariate
Gaussian, $X \indep Y$ is equivalent to
the true parameter $\Sigma_{12}$ being
zero, where $\Sigma$ is the unknown
covariance matrix of $(X,Y)$. In that
particular situation $H_0$ can be tested 
by computing the sample correlation 
coefficient of $(X_n,Y_n)$ and comparing 
it to its null distribution for that 
sample size.

However, in general we do not 
know whether data come from a Gaussian 
distribution, and the bivariate point 
cloud may have a different shape.
We have illustrated in Examples 1, 3, 
and 4 that dependent variables can be
uncorrelated, so a test of $\Cov(X,Y)$
would not suffice anyway.
What we need is a distribution-free
independence test, meaning that it
works for any distribution of $(X,Y)$.
Since we know that $X \indep Y$ is
characterized by $\dCov(X,Y)=0$, a natural
idea is to compute the test statistic
$\dCov(X_n,Y_n)$ from the sample. 
Larger values of $\dCov(X_n,Y_n)$ provide 
more evidence against $H_0$ than smaller
values, but how can we compute the
$p$-value when we do not know the kind 
of distribution that $(X,Y)$ has?

Since all we have is the dataset 
$(X_n,Y_n)$, this is what we must use.
Whatever the distribution of $(X,Y)$,
a random permutation of $Y_n$ will be
independent of $X_n$\,. More formally,
if we draw a permutation $\rp$ from
the uniform distribution on all $n!$
permutations on $(1,\ldots,n)$, we 
have $Y_n^{\rp} := (y_{\rp(1)},\ldots, 
y_{\rp(n)}) \indep X_n$\,.
If $n$ is very small we can use all
possible permutations $\rp$, and 
otherwise we can draw many of them, say 
$m=1000$ permutations 
$\rp_1,\ldots,\rp_m$\,. We can then
estimate the $p$-value by counting how
often $\dCov(X_n,Y_n^{\rp_m})$ with the
permuted $Y_n^{\rp_m}$ is larger than
the observed $\dCov(X_n,Y_n)$:
\begin{equation*}\label{permtest}
   \hat{p} = \frac{1}{m+1}
   \big(\#\big\{m\,\big|\, 
   \dCov(X_n,Y_n^{\rp_m})
   > \dCov(X_n,Y_n)\big\} + 1\big).
\end{equation*}
The $+1$ stems from the fact that the
original $Y_n$ corresponds to the identical
permutation $(1,\ldots,n)$ and is independent
of $X_n$ under $H_0$\,, and has the
advantage that $\hat{p}$ cannot become 
exactly zero, which would be unrealistic.

The permutation test 
is simple, and it is fast due to the fast 
algorithms for $\dCov$. Note that it would 
make no difference if we would replace 
$\dCov$ by $\dCor$, since the denominator
$(\dCov(X_n,X_n)\dCov(Y_n^{\rp_m},
Y^{\rp_m})_n)^{1/2} = (\dCov(X_n,X_n)
\dCov(Y_n,Y_n))^{1/2}$ of $\dCor$ is 
constant, so it is easiest to stick with 
$\dCov$. Also, it does not matter whether 
we square $\dCov$ or not.
More information about testing independence
can be found in \citep{SzekelyRizzo2023}.
A potential exercise for students would be 
to generate samples from the bivariate 
distributions in Example 3 or 4 of 
Section~\ref{sec:examples} and compute 
$\hat{p}$ for different sample sizes. In that 
setting they can also estimate the power of 
the permutation test for a fixed level, for 
instance by rejecting $H_0$ when 
$\hat{p} < 0.05$, using simulation.\\

\noindent{\bf Acknowledgment.} The authors thank
Gabor Sz\'ekely for a personal communication.\\

\noindent{\bf Software availability.} An
\textsf{R} script that reproduces the examples
in this note is available on
\url{https://wis.kuleuven.be/statdatascience/robust/software}\,.

\section*{Appendix with proofs}
\label{sec:proofs1}

In order to prove 
Proposition~\ref{prop:XminX'indepY}, it turns 
out that the following lemma is very helpful.

\begin{lemma} \label{lemma:XminX'indepY}
If $(X,Y)$ is a pair of random variables and we 
construct an independent copy $X'$ of $X$, that 
is, $X' \sim X$ and $X' \indep (X,Y)$, then
$(X-X') \indep Y$ is equivalent to the condition
\begin{equation} \label{eq:cond1}
  \mbox{for all $t$ with $\phi_{(X,Y)}(t,0) \neq 0:$ }
  \;\; \phi_{(X,Y)}(t,v)= 
    \phi_{(X,Y)}(t,0)\phi_{(X,Y)}(0,v)\;.
\end{equation}
\end{lemma}

\begin{proof}[Proof of Lemma~\ref{lemma:XminX'indepY}]
For the direction $\Rightarrow$ we compute the 
characteristic functions
\begin{align*}
  &\phi_{(X,-X',Y)}(s,t,v) = 
  \phi_{-X'}(t) \phi_{(X,Y)}(s,v)\\  
  &\phi_{(X-X',Y)}(t,v) = \phi_{X-X'}(t) \phi_{Y}(v)
  = \phi_X(t) \phi_{-X'}(t) \phi_Y(v)\;.
\end{align*}
On the subset $\{s=t\}$ both left hand sides equal
$\E[e^{it(X-X')+ivY}]$ so
\begin{equation}\label{eq:intermediate}
  \phi_{(X,Y)}(t,v) \phi_{X}(-t) = 
  \phi_X(t) \phi_{X}(-t) \phi_{Y}(v)\;.
\end{equation}
Since $\phi_X$ is Hermitian its set of roots is 
symmetric, so we have that 
$\phi_X(t) \neq 0 \Rightarrow \phi_{X}(-t) \neq 0$ 
and in that case $\phi_{X}(-t)$
cancels in~\eqref{eq:intermediate}, 
yielding~\eqref{eq:cond1}.

For the direction $\Leftarrow$ we compute
\begin{align*}
  \phi_{(X-X',Y)}(t,v) &= \E[e^{it(X-X')+ivY}] 
  = \E[e^{itX}e^{-itX'}e^{ivY}]\\
  &= \phi_{-X'}(t) \phi_{(X,Y)}(t,v) \quad
     \mbox{due to } X' \indep (X,Y)\;.
\end{align*}
In this equality we can replace $\phi_{(X,Y)}(t,v)$ by
$\phi_X(t) \phi_Y(v)$ whenever $\phi_X(-t) \neq 0$,
so then 
\begin{equation} \label{eq:factorizes}
  \phi_{(X-X',Y)}(t,v) = \phi_{X-X'}(t)\phi_Y(v)\,.
\end{equation}
But this also holds when $\phi_X(-t) = 0$ because then
$\phi_{-X'}(t) \phi_{(X,Y)}(t,v) = 0 =
\phi_{X}(-t) \phi_X(t) \phi_Y(v)$.
Therefore~\eqref{eq:factorizes} holds unconditionally,
hence $(X-X') \indep Y$.
\end{proof}

\vskip5mm

\begin{proof}[\bf Proof of Proposition 
\ref{prop:XminX'indepY}]
For (a) we use the fact that $X \indep Y$ 
implies $\phi_{(X,Y)}(t,v)=\phi_{X}(t)\phi_{Y}(v)$ 
for any $t$ and $v$, which is stronger than
condition~\eqref{eq:cond1} in
Lemma~\ref{lemma:XminX'indepY}, 
hence $(X-X') \indep Y$.

For (b) we also start from condition~\eqref{eq:cond1} 
in Lemma~\ref{lemma:XminX'indepY}.
If the characteristic function of $X$ has no 
roots we always have $\phi_{X}(-t) \neq 0$ so
\begin{equation} \label{eq:phifactorizes}
  \phi_{(X,Y)}(t,v) = \phi_X(t) \phi_{Y}(v) \quad
  \mbox{ for all } (t,v)
\end{equation}
hence $X \indep Y$.

Suppose that $\phi_X$ does have roots but they are 
isolated, implying that the non-roots form a dense set.
That is, any root $t$ is the limit of a sequence of 
non-roots $t_n$ for $n \rightarrow \infty$.
In each $t_n$ we have 
$\phi_{(X,Y)}(t_n,v)= \phi_{X}(t_n)\phi_{Y}(v)$
by condition~\eqref{eq:cond1}. Since characteristic
functions are absolutely continuous we can pass to the
limit, again yielding~\eqref{eq:phifactorizes}.

If we assume nothing about roots but $\phi_{(X,Y)}$ is 
analytic, so are $\phi_{X}(t) = \phi_{(X,Y)}(t,0)$ 
and $\phi_{Y}(v) = \phi_{(X,Y)}(0,v)$.
All characteristic functions take the value 1 at 
the origin, and are absolutely continuous.
Therefore there is a $\delta > 0$ such that for 
all $(t,v)$ in the disk $B((0,0),\delta)$ it 
holds that $\phi_{(X,Y)}(t,v)$ 
as well as $\phi_{X}(t)$ and $\phi_{Y}(v)$ are
nonzero. On that disk we can thus divide by
$\phi_{X}(-t)$ in~\eqref{eq:intermediate}, hence 
$\phi_{(X,Y)}(t,v) = \phi_X(t) \phi_{Y}(v)$
holds on it.
Since $\phi_{X}(t)$ and $\phi_{Y}(v)$ are 
analytic, so is their product. By analytic 
continuation~\eqref{eq:phifactorizes}
holds, so again $X \indep Y$. 
\end{proof}

\vskip5mm

We now consider pairwise differences of both variables 
$X$ and $Y$. This requires a second lemma.

\begin{lemma} \label{lemma:XminX'indepYminY'}
If $(X,Y)$ is a pair of random variables and we 
construct an independent copy $(X',Y')$ of it, 
that is, $(X',Y') \sim (X,Y)$ and 
$(X',Y') \indep (X,Y)$, then 
$(X-X') \indep (Y-Y')$ is equivalent to the condition 
\begin{equation} \label{eq:cond2}
  |\phi_{(X,Y)}(t,v)| = |\phi_{(X,Y)}(t,0)|\; 
  |\phi_{(X,Y)}(0,v)| \quad \mbox{ for all } (t,v)\;.
\end{equation}
\end{lemma}

\begin{proof}[Proof of 
Lemma~\ref{lemma:XminX'indepYminY'}] For the direction 
$\Rightarrow$ we compute the characteristic functions
\begin{align*}
  &\phi_{(X,-X',Y,-Y')}(s,t,u,v) 
     = \phi_{(X,Y)}(s,u) \phi_{(-X',-Y')}(t,v)
     = \phi_{(X,Y)}(s,u) \phi_{(X,Y)}(-t,-v)\\  
  &\phi_{(X-X',Y-Y')}(t,v) = \phi_{X-X'}(t) \phi_{Y-Y'}(v)
     = \phi_X(t) \phi_{-X'}(t) \phi_Y(v) \phi_{-Y'}(v)\;.
\end{align*}
On the subset $\{s=t, u=v\}$ both left hand sides equal
$\E[e^{it(X-X')+iv(Y-Y')}]$.
Therefore
\begin{equation*} \label{eq:equalmod}
 |\phi_{(X,Y)}(t,v)|^2 =
 \phi_{(X,Y)}(t,v) \phi_{(X,Y)}(-t,-v) = 
 \phi_{X-X'}(t) \phi_{Y-Y'}(v)
  = |\phi_X(t)|^2 |\phi_Y(v)|^2\;.
\end{equation*}

For the direction $\Leftarrow$ we compute
\begin{align*}
  \phi_{(X-X',Y-Y')}(t,v) 
  &= \E[e^{itX}e^{-itX'}e^{ivY}e^{-ivY'}]\\
  &= \E[e^{itX+ivY}]\; \E[e^{-itX'-ivY'}] 
    \quad \mbox{due to } (X',Y') \indep (X,Y)\\
  &= \phi_{(X,Y)}(t,v) \overline{\phi_{(X,Y)}(t,v)}
   = |\phi_{(X,Y)}(t,v)|^2\\
  &= |\phi_X(t)|^2\;|\phi_Y(v)|^2 \quad 
     \mbox{from~\eqref{eq:cond2}}\\
  &= \phi_{X-X'}(t) \; \phi_{Y-Y'}(v)
\end{align*}
hence $(X-X') \indep (Y-Y')$.
\end{proof}

\vskip5mm

\begin{proof}[\bf Proof of Proposition 
\ref{prop:XminX'indepYminY'}] For (a) we use the
fact that $X \indep Y$ implies that\linebreak
$\phi_{(X,Y)}(t,v)= \phi_{X}(t)\phi_{Y}(v)$ 
for any $t$ and $v$, hence 
$|\phi_{(X,Y)}(t,v)| = |\phi_X(t)|\; |\phi_{Y}(v)|$
which is condition~\eqref{eq:cond2} in
Lemma~\ref{lemma:XminX'indepYminY'}, so 
$(X-X') \indep (Y-Y')$.

For (b), $(X-X') \indep Y$ implies
   $\phi_{(X,Y)}(t,v)= \phi_{X}(t)\phi_{Y}(v)$
for all $t$ with $\phi_X(t) \neq 0$ by 
Lemma~\ref{lemma:XminX'indepY}. In the remaining
points $(t,v)$ it holds that $\phi_X(t) = 0$
and then 
$|\phi_{(X,Y)}(t,v)|= |\phi_{X}(t)|\;|\phi_{Y}(v)|
= 0$ by condition~\eqref{eq:cond2} of
Lemma~\ref{lemma:XminX'indepYminY'}, so 
$\phi_{(X,Y)}(t,v) = 0 = 
\phi_{X}(t)\phi_{Y}(v)$ as well. The combination
yields~\eqref{eq:phifactorizes}, hence $X \indep Y$.

Part (c).
By symmetry of $(X,Y)$ and hence of $X$ and $Y$
we know that $\phi_{(X,Y)}(t,v)$ as well as
$\phi_X(t)$ and $\phi_Y(v)$ are real and even, hence
condition~\eqref{eq:cond2} yields
\begin{equation} \label{eq:equalsquares}
  \phi_{(X,Y)}^2(t,v) = 
  \phi_{X}^2(t) \phi_{Y}^2(v)\;.
\end{equation}

If $\phi_{(X,Y)}(t,v)$ has no roots, it follows from
$\phi_{(X,Y)}(0,0)=1$ and continuity of $\phi_{(X,Y)}$
that always $\phi_{(X,Y)}(t,v) > 0$. Therefore also
$\phi_{X}(t) = \phi_{(X,Y)}(t,0) > 0$ and 
$\phi_{Y}(v) = \phi_{(X,Y)}(0,v) > 0$.
Taking square roots on both sides 
of~\eqref{eq:equalsquares} 
yields~\eqref{eq:phifactorizes}, hence $X \indep Y$.

If, on the other hand, $\phi_{(X,Y)}$ is analytic, 
so are $\phi_{X}(t) = \phi_{(X,Y)}(t,0)$ and 
$\phi_{Y}(v) = \phi_{(X,Y)}(0,v)$.
All characteristic functions take the value 1 at 
the origin, and are absolutely continuous.
Therefore there is a $\delta > 0$ such that for all 
$(t,v)$ in the disk $B((0,0),\delta)$ it holds that 
$\phi_{(X,Y)}(t,v)$ as well as $\phi_{X}(t)$ and 
$\phi_{Y}(v)$ are strictly positive. 
On that disk we can thus take
square roots of~\eqref{eq:equalsquares}, yielding 
$\phi_{(X,Y)}(t,v) = \phi_X(t) \phi_{Y}(v)$ on it. 
Since $\phi_{X}(t)$ and $\phi_{Y}(v)$ are analytic, 
so is their product. By analytic continuation the 
equality must hold everywhere, 
yielding~\eqref{eq:phifactorizes} so again $X \indep Y$. 
\end{proof}

% \bibliographystyle{chicago}
% \bibliography{bibliography}

\end{document}